\documentclass[11pt]{article}
\usepackage[margin=1in]{geometry}

\usepackage{amsmath, amssymb, amsthm}
\usepackage[mathic]{mathtools}

\usepackage[T1]{fontenc}
\usepackage{libertine}
\usepackage[libertine]{newtxmath} 
\usepackage[scaled=0.96]{zi4} 

\usepackage{a4wide, dsfont, microtype, xcolor, paralist}

\usepackage{subcaption}

\usepackage[ocgcolorlinks]{hyperref} 
\colorlet{DarkRed}{red!50!black}
\colorlet{DarkGreen}{green!50!black}
\colorlet{DarkBlue}{blue!50!black}
\hypersetup{
	linkcolor = DarkRed,
	citecolor = DarkGreen,
	urlcolor = DarkBlue,
	bookmarksnumbered = true,
	linktocpage = true
}

\newcommand{\DD}[0]{\ensuremath{\mathtt{D}}}
\newcommand{\LL}[0]{\ensuremath{\mathtt{L}}}
\newcommand{\wc}[0]{wc}

\newcommand{\BWT}[0]{\ensuremath{\mathtt{XBWT}}}

\newcommand{\select}[0]{\textit{select}}
\newcommand{\prank}[0]{\textit{p\_rank}}
\newcommand{\rank}[0]{\textit{rank}}

\DeclareMathOperator{\onesop}{\mathtt{ones}}

\usepackage{float}

\usepackage{thmtools} 

\usepackage{array}
\newcolumntype{P}[1]{>{\footnotesize\centering\arraybackslash}p{#1}}

\usepackage{tikz}
\usetikzlibrary{positioning}
\usetikzlibrary{arrows.meta}
\usetikzlibrary{automata}
\usetikzlibrary{matrix}
\usetikzlibrary{calc}

\DeclareMathOperator{\E}{E}

\let\epsilon\varepsilon

\usepackage{xspace}

\usepackage[linesnumbered,ruled,vlined,boxed]{algorithm2e}
\newlength{\commentWidth}
\let\oldnl\nl
\newcommand{\nonl}{\renewcommand{\nl}{\let\nl\oldnl}}
\DontPrintSemicolon
\SetKwInOut{Input}{Input}\SetKwInOut{Output}{Output}

\usepackage{tikz}
\usetikzlibrary{decorations.pathreplacing}
\usetikzlibrary{plotmarks}
\usetikzlibrary{positioning,automata,arrows}
\usetikzlibrary{shapes.geometric}
\usetikzlibrary{decorations.markings}
\usetikzlibrary{positioning, shapes, arrows}

\PassOptionsToPackage{usenames,dvipsnames,svgnames}{xcolor}
\definecolor{orange}{RGB}{235,90,0}
\definecolor{darkorange}{RGB}{175,30,0}
\definecolor{turkis}{RGB}{131,182,182}
\definecolor{darkturkis}{RGB}{31,82,82}
\definecolor{green}{RGB}{102,180,0}
\definecolor{darkgreen}{RGB}{51,90,0}
\definecolor{myblue}{RGB}{0,0,213}
\definecolor{mydarkblue}{RGB}{0,0,100}
\definecolor{mybrightblue}{HTML}{74B0E4}
\definecolor{mybrighterblue}{HTML}{B3EAFA}
\definecolor{lila}{RGB}{102,0,102}
\definecolor{darkred}{RGB}{139,0,0}
\definecolor{darkyellow}{RGB}{188,135,2}
\definecolor{brightgray}{RGB}{200,200,200}
\definecolor{darkgray}{RGB}{50,50,50}
\definecolor{amaranth}{rgb}{0.9, 0.17, 0.31}
\definecolor{alizarin}{rgb}{0.82, 0.1, 0.26}
\definecolor{amber}{rgb}{1.0, 0.75, 0.0}
\definecolor{green(ryb)}{rgb}{0.4, 0.69, 0.2}
\definecolor{hanblue}{rgb}{0.27, 0.42, 0.81}
\definecolor{grannysmithapple}{rgb}{0.66, 0.89, 0.63}

\newtheorem{theorem}{Theorem}[section]
\newtheorem{lemma}[theorem]{Lemma}
\newtheorem{definition}[theorem]{Definition}
\newtheorem{corollary}[theorem]{Corollary}

\newenvironment{acknowledgements}{%
  \begin{abstract}
}{%
  \end{abstract}
}

\title{Indexing Tries within Entropy-Bounded Space}
\author{
 	Lorenzo Carfagna\thanks{Department of Computer Science, University of Pisa, Italy} \\ \texttt{lorenzo.carfagna@phd.unipi.it} \and
    Carlo Tosoni\thanks{DAIS, Ca' Foscari University of Venice, Italy} \\
    \texttt{carlo.tosoni@unive.it}
}
\date{}

\begin{document}
\maketitle

\begin{abstract}
We study the problem of indexing and compressing tries using a BWT-based approach.
Specifically, we consider a succinct and compressed representation of the XBWT of Ferragina et al.\ [FOCS '05, JACM '09] corresponding to the analogous of the FM-index [FOCS '00, JACM '05] for tries.
This representation allows to efficiently count the number of nodes reached by a given string pattern.
To analyze the space complexity of the above trie index, we propose a proof for the combinatorial problem of counting the number of tries with a given symbol distribution.
We use this formula to define a worst-case entropy measure for tries, as well as a notion of $k$-th order empirical entropy.
In particular, we show that the relationships between these two entropy measures are similar to those between the corresponding well-known measures for strings.
We use these measures to prove that the XBWT of a trie can be encoded within a space bounded by our $k$-th order empirical entropy plus a $o(n)$ term, with $n$ being the number of nodes in the trie.
Notably, as happens for strings, this space bound can be reached for every sufficiently small $k$ simultaneously.
Finally, we compare the space complexity of the above index with that of the $r$-index for tries proposed by Prezza [SODA '21] and we prove that in some cases the FM-index for tries is asymptotically smaller.
\end{abstract}

\begin{acknowledgements}
We would like to thank Sung-Hwan Kim for remarkable suggestions concerning the combinatorial problem addressed in the article.

\emph{Lorenzo Carfagna}: partially funded by the INdAM-GNCS Project CUP \\E53C24001950001, 
and by the PNRR ECS00000017 Tuscany Health Ecosystem, Spoke 6, CUP I53C22000780001, funded by the NextGeneration EU programme.

\emph{Carlo Tosoni}: funded by the European Union (ERC, REGINDEX, 101039208). Views and opinions expressed are however those of the authors only and do not necessarily reflect those of the European Union or the European Research Council Executive Agency. Neither the European Union nor the granting authority can be held responsible for them.
\end{acknowledgements}

\thispagestyle{empty}
\pagebreak

\section{Introduction}

The Burrows-Wheeler transform~\cite{BWT} (BWT) is a renowned reversible string permutation which rearranges the characters based on the lexicographic order of the suffixes following them.
This transform enhances the compressibility of the input string and at the same time it allows to efficiently answer \emph{count queries}, i.e., to return the number of occurrences of a pattern in the original string~\cite{FMindex}.
Later, extensions of the BWT were developed for more complex objects such as ordered labeled trees~\cite{xBWTjournal}, de Bruijn graphs~\cite{deBruijnBWT}, and finite-state automata~\cite{WheelerGraphs, BWTnfas}.
All these extensions rely on the key idea of sorting the nodes according to the order of their incoming strings to achieve compression and to support indexing features such as count queries.
In particular, the original BWT for strings is known to be correlated with the notion of \emph{empirical entropy}~\cite{BWTanalysis}, since the BWT can be compressed in a space close to the $k$-th order empirical entropy of the input string. 
This result led to the development of some of the most famous compressed string indices, whose space occupation is proportional to the $k$-th order entropy of the string~\cite{FMindex,FMI-Journal,FullTextIndexes,FullTextIndexesSurvey}.
The importance of these results stems from the fact that in some applications, such as bioinformatics, strings tend to be particularly repetitive and consequently the $k$-th order entropy becomes low.
A similar connection between the BWT and a different notion of $k$-th order entropy is also known for labeled ordered trees.
Indeed, the XBWT of Ferragina et al.~\cite{xBWTconf, xBWTjournal} linearises the node labels of the input tree, which can then be compressed to the $k$-th order label entropy.
In this case, the $k$-length context of a label is defined as the $k$-length string entering into the corresponding tree node.
However, to reconstruct the input tree it is also necessary to store the tree topology separately.
At this point, it is natural to ask whether similar results can be established also for other types of objects.
In particular, in this paper we study this problem for the more specific case of tries, and consequently we aim to relate the XBWT of a trie with some notion of trie entropy.
For tries there already exists a notion of worst-case entropy $\mathcal{C}(n,\sigma)$~\cite{RepresentingTrees,RRR} equal to $\log_2 \mathcal{S}(n,\sigma)$, where $\mathcal{S}(n,\sigma)=\frac{1}{n}\binom{\sigma n}{n-1}$ is the number of tries with $n$ nodes over an alphabet $\Sigma$ of size $\sigma$~\cite{concreteMathematics}.
In the literature, there are many trie representations whose space usage is expressed in terms of $\mathcal{C}(n,\sigma)$.
For instance, Benoit et al.\ proposed a trie representation taking $\mathcal{C}(n,\sigma) + \Omega(n)$ bits~\cite{RepresentingTrees}.
This space occupation was later improved to $\mathcal{C}(n, \sigma) + o(n) + O(\log \log \sigma)$ bits by Raman et al.~\cite{RRR} and Farzan et al.~\cite{UniversalSuccinctTree?}
All these data structures support queries on the trie topology plus the cardinal query of retrieving the child labeled with the $i$-th character of the alphabet in constant time.
However, these data structures do not support count queries, which in the context of tries corresponds to counting the number of nodes reached by an input string pattern.
For the sake of completeness, we also mention solutions proposing dynamic representation of tries with space complexity $\mathcal{C}(n,\sigma) + \Omega(n)$~\cite{ArroyueloDynamicTries, Bonsai, DavoodiDynamicTries}.

Similarly to what has been done for strings, in this paper we aim to refine the above approach, and develop an entropy formula, and subsequently a trie index, which contrary to $\mathcal{C}(n,\sigma)$ takes into account also the distribution of the edge labels.

\paragraph{Our contribution.} Motivated by the above reasons, in this paper we study the combinatorial problem of finding the number $\lvert \mathcal{U} \rvert$ of tries given an input \emph{symbol distribution} $\{n_c\mid c\in \Sigma\}$, that is tries having $n_c$ edges labeled with symbol $c$.
During a bibliographic search, we found a recent technical report~\cite{TRProdinger} showing $\lvert \mathcal{U}\rvert = \frac{1}{n}\prod_{c \in \Sigma}\binom{n}{n_c}$, with $n$ being the number of nodes, by sketching a proof based on generating functions~\cite{analyticCombinatorics}.
However, in this paper we prove this closed formula by showing a simple bijection between these tries and a particular class of binary matrices.
As an immediate consequence, we obtain a worst-case entropy $\mathcal{H}^{\wc}$ for this refined class of tries.
In addition, we define a notion of $k$-th order empirical entropy $\mathcal{H}_k$ for tries.
Similarly to the label entropy of the XBWT~\cite{xBWTconf,xBWTjournal}, our $\mathcal{H}_k$ divides the nodes based on their $k$-length incoming string, but it also encodes the topology of the trie.
Moreover, we show that $\mathcal{H}^{\wc}$ and $\mathcal{H}_{k}$ share similar properties of their corresponding string counterparts.
In particular, it is $n\mathcal{H}_0=\mathcal{H}^{\wc}+O(\sigma \log n)$ and for every $k \geq 0$ it holds that $\mathcal{H}_{k+1} \leq \mathcal{H}_{k}$.
Next, we show the existence of an XBWT representation for tries which, analogously to the FM-index~\cite{FMindex, FMI-Journal} for strings, efficiently supports count queries within $n\mathcal{H}_k + o(n)$ bits of space for every $k = o(\log_\sigma n)$ simultaneously assuming $\sigma \leq n^{\varepsilon}$ for some constant $0< \varepsilon < 1$.
In the case of large alphabets, this representation coincides with the solution proposed by Kosolobov,  Sivukhin~\cite[Lemma 6]{Kosolobov2019}, and Belazzougui~\cite[Theorem 2]{Belazzougui2010} for representing the trie underlying the Aho-Corasick automaton~\cite{AC75}.
For small alphabets, we propose an alternative representation to speed up the execution time of count queries.
While previous works~\cite{Belazzougui2010,Hon2013,Kosolobov2019} analyze the space complexity of this data structure using the label entropy of Ferragina et al.~\cite{xBWTconf}, we give a space analysis based on our entropy measure for tries.
Moreover, we show that if for every $n_c$ it holds that $n_c \leq n /2$ then this representation is \emph{succinct}, i.e., it takes $\mathcal{H}^{\wc}+o(\mathcal{H}^{\wc})$ bits of space.
Finally, we compare the above index with the (trie) $r$-index of Prezza~\cite{rindexTrie} which can be stored within $O(r \log n) + o(n)$ bits of space, with $r$ being the number of runs in the XBWT of the trie.
We show that, due to~\cite[Theorem 3.1]{rindexTrie} for every trie and integer $k \geq 0$ it holds that $r \leq n\mathcal{H}_k + \sigma^{k+1}$.
Moreover, we prove that there exists an infinite family of tries for which $r = \Theta(n\mathcal{H}_0) = \Theta(n)$ holds.
This demonstrates that the FM-index for tries can be asymptotically smaller with respect to the $r$-index.

\section{Notation}\label{sec: Notation}
In the following we refer with $\Sigma$ a finite alphabet of size $\sigma$, totally ordered according to a given relation $\preceq$.
For every integer $k\geq 0$ we define $\Sigma^k$ as the set of all length-$k$ strings with symbols in $\Sigma$ where $\Sigma^0=\{\epsilon\}$, i.e., the singleton set containing the empty string $\epsilon$.
Moreover, we denote by $\Sigma^*$ the set of all finite strings over the alphabet $\Sigma$, i.e., the union of all $\Sigma^k$ for every $k\geq 0$.
We extend $\preceq$ \emph{co-lexicographically} to the set $\Sigma^*$, i.e., by comparing the  characters from right to left and considering $\epsilon$ as the smallest string in $\Sigma^*$.
With $[n]$ we denote the set $\{1,..,n\}\subseteq \mathbb{N}$ and given a set $X$, $\lvert X \rvert$ is the cardinality of $X$.
Unless otherwise specified, all logarithms are base $2$ and denoted by $\log x$, furthermore we assume 
that $0\log(x/0)=0$ for every $x \geq 0$.
Given a matrix $M$ we denote the $i$-th row/column of $M$ as $M[i][-]$ and $M[-][i]$, respectively.
In the following, we denote with $\onesop(M')$ the number of entries equal to $1$ in a submatrix $M'$ of $M$.
We work in the RAM model, where we assume that the word size is $\Theta(\log n)$, where $n$ is the input dimension.

In this paper, we work with \emph{cardinal trees}, also termed tries, i.e., edge-labeled ordered trees in which (i)~the labels of the edges outgoing from the same node are distinct, and (ii)~sibling nodes are ordered by their incoming label.
We denote by $\mathcal{T} = (V,E)$ a trie over an alphabet $\Sigma$, where $V$ is set of nodes with $\lvert V\rvert=n$ and $E$ is the set of edges with labels drawn from $\Sigma$.
We denote with $n_i$ the number of edges labeled by $i$-th character $c_i \in \Sigma$ according to $\preceq$.
Note that since $\mathcal{T}$ is a tree we have that $\sum_{i=1}^{\sigma} n_i = \lvert E \rvert= n-1$.
We denote by $\lambda(u)$ the label of the incoming edge of a node $u \in V$. 
If $u$ is the root, we define $\lambda(u) = \#$, where $\#$ is the smallest character of $\Sigma$ according to $\preceq$ not labeling any edge.
Given a node $u \in V$, the function $\pi(u)$ returns the parent node of $u$ in $\mathcal{T}$, where $\pi(u) = u$ if $u$ is the root.
Furthermore, the function $out(u)$ returns the set of labels of the edges outgoing from node $u$, i.e., $out(u) = \{c \in \Sigma: (u,v,c) \in E\; \text{for some}\; v \in V \}$.


Given a bitvector $B$ of size $m$ with $n$ entries equal to $1$, in the following we denote with $\rank(i,B)$ the number of $1$s in $B$ up to position $i\leq m$ (included), and with $\select(i,B)$ the position of the $i$-th occurrence of $1$ in $B$ where $1\leq i\leq n$.
The \emph{indexable dictionary} (ID) problem consists of representing $B$ and supporting \emph{select} and \emph{partial rank} queries in constant time.
In particular, given an integer $i \in [m]$, a partial rank query $\prank(i, B)$ returns $-1$ if $B[i] = 0$ and $\rank(i,B)$ otherwise.
Note that both rank and partial rank queries can be used to determine if $B[i]=1$.
On the other hand, if a representation supports (full) rank and select queries on both $B$ and its complementary bitvector $\bar B$, then it is called a \emph{fully} indexable dictionary (FID) representation.
We note that given an ID representation of a bitvector $B$ of size $m$ we can perform rank queries $\rank(i,B)$ in $O(\log m)$ time when $B[i]=0$ by performing a binary search using $\select$ queries.
In this paper, we use the ID and FID representations proposed by Raman et al.~\cite{RRR} that are summarized in the following lemma.

\begin{lemma}\cite[Lemma 4.1 and Theorem 4.6]{RRR}\label{thm:RRR}
Given an arbitrary bitvector $B$ of size $m$ with $n$ ones, there exist:
\begin{enumerate}
    \item\label{thm:RRR property 1} a FID representation of $B$ taking $\lceil \log \binom{m}{n}\rceil + O(m \log\log m / \log m)$ bits.
    \item\label{thm:RRR property 2} an ID representation of $B$ taking $\lceil \log \binom{m}{n}\rceil + o(n) + O(\log \log m)$ bits.
\end{enumerate}
\end{lemma}

\section{On the number of tries with a given symbol distribution}

This section is dedicated to prove the following result.

\begin{theorem}\label{theorem number of tries}
Let $\mathcal{U}$ be the set of tries with $n$ nodes and labels drawn from an alphabet $\Sigma=\{c_1,..,c_\sigma\}$, where each symbol $c_i$ labels $n_i$ edges, then $|\mathcal{U}| = \frac{1}{n}\prod_{i=1}^{\sigma} \binom{n}{n_i}$
\end{theorem}

The same problem has already been addressed by Prezza (see Section~3.1 of~\cite{rindexTrie}) and Prodinger~\cite{TRProdinger}.
Prezza claimed that $\lvert \mathcal{U} \rvert$ is equal to $\prod_{i=1}^{\sigma}\binom{n}{n_i}$. However, we point out that this formula overestimates the correct number of tries by a multiplicative factor $n$.
On the other hand, Prodinger deduced the correct number in a technical report using the Lagrange Inversion Theorem~\cite{analyticCombinatorics}.
In this paper, we give an alternative proof using a simple bijection between these tries and a class of binary matrices. 
For the more general problem of counting tries with $n$ nodes and $\sigma$ characters, the corresponding formula $\mathcal{S}(n,\sigma)$ is well-known and in particular it is $\mathcal{S}(n,\sigma)=\frac{1}{n}\binom{n\sigma}{n-1}$~\cite{concreteMathematics, cycleLemma}. 
We now introduce the definition of set $\mathcal{M}$, which depends on $\mathcal{U}$.

\begin{definition}[Set $\mathcal{M}$]\label{Set M}
    We define $\mathcal{M}$ as the set of all $\sigma \times n$ binary matrices $M$, such that $\onesop(M[i][-]) = n_i$ for every $i \in [\sigma]$.
\end{definition}

Note that $\onesop(M) = n - 1$, trivially follows from $\sum_{i=1}^{\sigma} n_i = n-1$.
In order to count the elements in the set $\mathcal U$, we now define a function $f$ mapping each trie $\mathcal{T}$ into an element of $M \in \mathcal{M}$.
According to our function $f$, the number of ones in the columns of $M$ encodes the topology of the trie, while the fact that $M$ is binary encodes the standard trie labeling constraint. 

\begin{definition}[Function $f$]
    We define the function $f : \mathcal{U} \rightarrow \mathcal{M}$ as follows: given a trie $\mathcal{T} \in \mathcal U$ consider its nodes $u_{1}, \ldots, u_{n}$ in pre-order visit.
    Then $M = f(\mathcal{T})$ is the $\sigma \times n$ binary matrix such that $M[i][j] = 1$ if and only if $c_{i} \in out(u_{j})$.
\end{definition}

Observe that, given a trie $\mathcal{T}$, the $i$-th column $M[-][i]$ of the corresponding matrix $M=f(\mathcal{T})$ is the characteristic bitvector of the outgoing labels of node $u_i$ (under the specific ordering on $\Sigma$) and in particular it holds that $\onesop(M[-][i])$ is the out-degree of $u_i$.
It is easy to observe that the function $f$ is injective, however, $f$ is not surjective in general: some matrices $M \in \mathcal{M}$ may not belong to the image of $f$ because during the inversion, connectivity constraints could be violated, as shown in Figure~\ref{fig1}.
To characterize the image of $f$, we define the following two sequences.

\begin{definition}[Sequences $\DD$ and $\LL$]\label{arrays D and L}
    For every $M \in \mathcal{M}$ we define the integer sequences $\DD$ and $\LL$ of length $n$, such that $\DD[i] = \onesop(M[-][i]) - 1$ and $\LL[i] = \sum_{j = 1}^{i} \DD[j]$ for every $i \in [n]$.
\end{definition}

We can observe that if $M \in f(\mathcal{U})$, then if we consider the unique trie $\mathcal{T}=f^{-1}(M)$, $\DD$ is the out-degree of the $i$-th node in pre-order visit minus one, namely $\DD[i] = \rvert out(u_{i})\lvert - 1$.
Consequently, we have that $\LL[i] = \sum_{j = 1}^{i}\lvert out(u_j) \rvert - i$, and therefore $\LL[i] + 1$ denotes the total number of ``pending'' edges immediately after the pre-order visit of the node $u_i$, where the $+1$ term stems from the fact that the root has no incoming edges.
In other words, $\LL[i] + 1$ denotes the number of edges whose source has already been visited, while their destination has not.
We note that for every $M \in \mathcal{M}$, it holds that $\LL[n] = -1$, since $\onesop(M) = n - 1$.
In the following, we show that the sequence $\LL$ can be used to determine if $M \in f(\mathcal{U})$ holds for a given matrix $M \in \mathcal{M}$.
In fact, we will see that if $M \in f(\mathcal{U})$, the array $\LL$ is a so-called \emph{Łukasiewicz path} as defined in the book of Flajolet and Sedgewick~\cite{analyticCombinatorics} (see Section~I. 5. ``Tree Structures'').
Specifically, a Łukasiewicz path $\mathcal{L}$ is a sequence of $n$ integers satisfying the following conditions.
(i)~$\mathcal{L}[n] = -1$, and for every $i \in [n-1]$, it holds that (ii)~$\mathcal{L}[i]\geq 0$, and (iii)~$\mathcal{L}[i+1] - \mathcal{L}[i] \geq -1$.
Łukasiewicz paths can be used to encode the topology of a trie since it is known that their are in bijection with unlabeled ordered trees of $n$ nodes~\cite{analyticCombinatorics}.
Specifically, given an ordered unlabeled tree $\mathcal{T}'$ the $i$-th point in its corresponding Łukasiewicz path $\mathcal{L}$ is $\mathcal{L}[i]= \sum_{j = 1}^{i} (\lvert out(u_j)\rvert -1)$, i.e., $\mathcal{L}$ is obtained by prefix-summing the sequence formed by the nodes out-degree minus $1$ in pre-order.
The reverse process consists in scanning $\mathcal{L}$ left-to-right and appending the current node $u_i$ of out-degree $\mathcal{L}[i]-\mathcal{L}[i-1]+1$ to the deepest pending edge on the left; obviously $u_1$ is the root of $\mathcal{T}'$ where we assume $\mathcal{L}[0]=0$.
The next result states that the tries with a fixed number of occurrences of symbols are in bijection with the matrices of $\mathcal{M}$ whose corresponding array $\LL$ is a \emph{Łukasiewicz path}.
Similar results have been proved in other articles~\cite{aProblemOfArrangements, cycleLemma, roteCountingDegrees}.

\begin{restatable}[]{lemma}{lemmai}\label{characterizing f}
A binary matrix $M \in \mathcal{M}$ belongs to $f(\mathcal{U})$ if and only if, the corresponding array $\LL$ is a Łukasiewicz path.
\end{restatable}

\begin{proof}
$(\Rightarrow):$ Consider the trie $\mathcal{T} = (V,E) = f^{-1}(M)$ and its underlying unlabeled ordered tree $\mathcal{T}'$, i.e., $\mathcal{T}'$ is the ordered tree obtained by deleting the labels from $\mathcal{T}$.
Obviously, the corresponding nodes of $\mathcal{T}$ and $\mathcal{T}'$ in pre-order have the same out-degree. Therefore, by Definition~\ref{arrays D and L}, the array $\LL$ of $M$ is equal to the Łukasiewicz path of $\mathcal{T}'$.

$(\Leftarrow):$ 
Given $M\in \mathcal{M}$ and its corresponding array $\LL$, we know by hypothesis that $\LL$ is a Łukasiewicz path. Therefore, we consider the unlabeled ordered tree $\mathcal{T}'$ obtained by inverting $\LL$.
Let $u_i$ be the $i$-th node of $\mathcal{T}'$ in pre-order, we label the $j$-th left-to-right edge outgoing from $u_i$ with the symbol corresponding to the $j$-th top-down $1$ in $M[-][i]$. The resulting order labeled tree is a trie $\mathcal{T}''$ and by construction it follows that $f(\mathcal{T}'')=M$.
\end{proof}

To count the number of tries in $\mathcal{U}$, we now introduce the concept of \emph{rotations}.

\begin{definition}[Rotations]\label{def rotations}
    For every matrix $M \in \mathcal{M}$ and integer $r \geq 0$, we define the $r$-th rotation of $M$, denoted by $M^{r}$, as $M^{r}[-][j] = M[-][(j+r-1)\bmod{n} +1]$ for every $j \in [n]$.
\end{definition}

Informally, a rotation $M^{r}$ of a matrix $M \in \mathcal{M}$ is obtained by moving the last $r$ columns of $M$ at its beginning, and therefore $M^0=M$. Moreover, since $M^r = M^{r\bmod {n}}$, in the following we consider only rotations $M^r$ with $r\in [0,n-1]$.
Note also that, for every $r \geq 0$, we know that $M^r \in \mathcal{M}$ as rotations do not change the number of entries equal to $1$ in a row. 
Moreover, we can observe that the array $\DD$ of a rotation $M^r$ corresponds to a cyclic permutation of the array $\DD'$ of $M$.
Next we show that all $n$ rotations of a matrix $M\in \mathcal{M}$ are distinct and exactly one rotation has an array $\LL$ which is Łukasiewicz.
To do this, we recall a known result from~\cite{roteCountingDegrees} and~\cite[Section~I. 5. ``Tree Structures'', Note I.47]{analyticCombinatorics}.

\begin{lemma}\cite[Lemma 2]{roteCountingDegrees}\label{rotation principle}
    Consider a sequence of integers $A = a_1, \ldots, a_n$ with $\sum_{i=1}^{n}a_i = -1$. Then (i)~all $n$ cyclic permutations of $A$ are distinct, and (ii)~there exists a unique cyclic permutation $A' = a_1', \ldots , a_n'$ of $A$ such that $\sum_{i=1}^{j}a_i' \geq 0$, for every $j \in [n-1]$
\end{lemma}

The following corollary is a consequence of the above lemma and an example is shown in Figure~\ref{fig1}.

\begin{restatable}[]{corollary}{corollaryi}\label{unique matrix in image}
    For every matrix $M \in \mathcal{M}$ all its rotations are distinct and only one of them belongs to $f(\mathcal{U)}$.
\end{restatable}

\begin{proof}
    Consider the sequence $\DD$ of matrix $M$. 
    We know that the sequence of integers $\DD[1], \ldots , \DD[n]$ sums to $-1$.
    Consequently, by Lemma~\ref{rotation principle}, it follows that all the cyclic permutations of $\DD$ are distinct and thus all the corresponding rotations of $M$ are distinct too.
    By Lemma~\ref{characterizing f}, to conclude the proof it remains to be shown that there exists a unique rotation $M^r$ whose array $\LL$ is a Łukasiewicz path.
    By Lemma~\ref{rotation principle}, we know there exists a unique $M^r$ rotation for which its corresponding array $\LL$ is nonnegative excluding the last position, which is a necessary condition for $\LL$ to be a Łukasiewicz path.
    Therefore $M^r$ is the only candidate that can belong to $f(\mathcal{U})$.
    Moreover, we know that $\LL[n] =-1$ and by Definition~\ref{arrays D and L} it holds that $\LL[i+1]-\LL[i]=\onesop(M^r[-][i])-1 \geq -1$ for every $i \in [n-1]$. 
    Since these observations prove that $\LL$ is Łukasiewicz, it follows that $M^r \in f(\mathcal{U})$.
\end{proof}

\makeatletter
\newcommand\notsotiny{\@setfontsize\notsotiny\@vipt\@viipt}
\makeatother

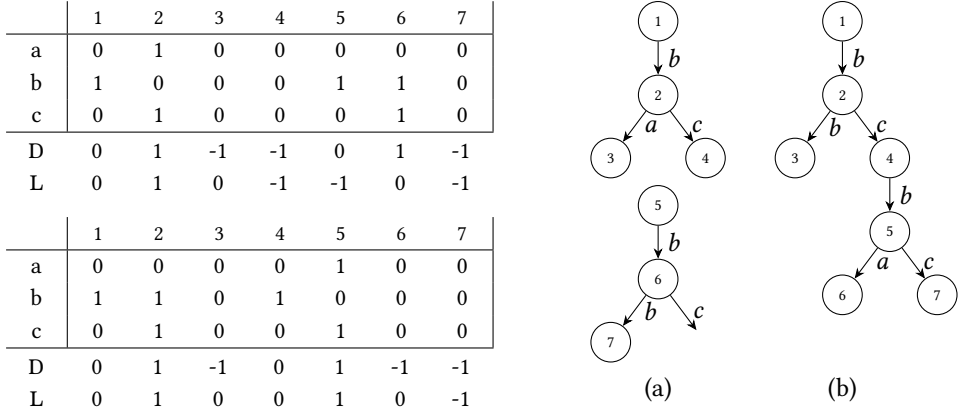
\begin{figure}[H]
	\centering
	\renewcommand{\arraystretch}{0.9}
	\begin{subfigure}{0.45\textwidth}
    \vspace{0pt}
	\begin{tabular}{ P{1em} | P{1em} P{1em} P{1em} P{1em} P{1em} P{1em} P{1em} |}
		   & \scriptsize{1} & \scriptsize{2} & \scriptsize{3} & \scriptsize{4} & \scriptsize{5} & \scriptsize{6} & \scriptsize{7} \\
		 \hline
		 \footnotesize{a} & 0 & 1 & 0 & 0 & 0 & 0 & 0 \\
		 \footnotesize{b} & 1 & 0 & 0 & 0 & 1 & 1 & 0 \\
		 \footnotesize{c} & 0 & 1 & 0 & 0 & 0 & 1 & 0 \\
		 \hline

	\end{tabular}
		
	\begin{tabular}{ P{1em} P{1em} P{1em} P{1em} P{1em} P{1em} P{1em} P{1em} }

		  $\DD$ & 0 & 1 & -1 & -1 & 0 & 1 & -1 \\
		$\LL$ & 0 & 1 & 0 & -1 & -1 & 0 & -1 \\

	\end{tabular}

	\vspace{0.5em}
	
	\begin{tabular}{ P{1em} | P{1em} P{1em} P{1em} P{1em} P{1em} P{1em} P{1em} |}
	& \scriptsize{1} & \scriptsize{2} & \scriptsize{3} & \scriptsize{4} & \scriptsize{5} & \scriptsize{6} & \scriptsize{7} \\
	\hline
	\footnotesize{a} & 0 & 0 & 0 & 0 & 1 & 0 & 0 \\
	\footnotesize{b} & 1 & 1 & 0 & 1 & 0 & 0 & 0 \\
	\footnotesize{c} & 0 & 1 & 0 & 0 & 1 & 0 & 0 \\
	\hline
	
	\end{tabular}
		
	\begin{tabular}{ P{1em} P{1em} P{1em} P{1em} P{1em} P{1em} P{1em} P{1em} }
		
		$\DD$ & 0 & 1 & -1 & 0 & 1 & -1 & -1 \\
		$\LL$ & 0 & 1 & 0 & 0 & 1 & 0 & -1 \\
		
	\end{tabular}

	\end{subfigure}
	\begin{subfigure}{0.4\textwidth}
	\centering
    \raisebox{-0.7em}{
	\begin{tikzpicture}[
		dim/.style={minimum size=1em, font={\notsotiny}}, 
		scale = 0.7,
		ghost/.style={draw=none},
        writes/.style={font={\small}},
		dots/.style={text centered, font={\small}}]

		\node[state, dim] (1) at (0,0) {1};
		\node[state, dim] (2) at (0,-1.4) {2};
		\node[state, dim] (3) at (-0.9,-2.6) {3};
		\node[state, dim] (4) at (0.9,-2.6) {4};
		\node[state, dim] (5) at (0,-3.5) {5};
		\node[state, dim] (6) at (0,-4.9) {6};
		\node[state, dim] (7) at (-0.9,-6.1) {7};
		\node[state, dim, ghost] (8) at (0.9,-6.1) {};

        \node[dots] (9) at (0, -7) {(a)};
		
		\draw[-{Stealth[length=1.4mm, width=1.2mm]}] (1) to node [right, writes] {$b$} (2);
		\draw[-{Stealth[length=1.4mm, width=1.2mm]}] (2) to node [right, writes] {$a$} (3);
		\draw[-{Stealth[length=1.4mm, width=1.2mm]}] (2) to node [right, writes] {$c$} (4);
		\draw[-{Stealth[length=1.4mm, width=1.2mm]}] (5) to node [right, writes] {$b$} (6);
		\draw[-{Stealth[length=1.4mm, width=1.2mm]}] (6) to node [right, writes] {$b$} (7);
		\draw[-{Stealth[length=1.4mm, width=1.2mm]}] (6) to node [right, writes] {$c$} (8);

        \begin{scope}[xshift=3.5cm]

		\node[state, dim] (1) at (0,0) {1};
		\node[state, dim] (2) at (0,-1.4) {2};
		\node[state, dim] (3) at (-0.9,-2.6) {3};
		\node[state, dim] (4) at (0.9,-2.6) {4};
		\node[state, dim] (5) at (0.9,-4.0) {5};
		\node[state, dim] (6) at (0,-5.2) {6};
		\node[state, dim] (7) at (1.8,-5.2) {7};
			

        \node[dots] (9) at (0, -7) {(b)};
            
		\draw[-{Stealth[length=1.4mm, width=1.2mm]}] (1) to node [right, writes] {$b$} (2);
		\draw[-{Stealth[length=1.4mm, width=1.2mm]}] (2) to node [right, writes] {$b$} (3);
		\draw[-{Stealth[length=1.4mm, width=1.2mm]}] (2) to node [right, writes] {$c$} (4);
		\draw[-{Stealth[length=1.4mm, width=1.2mm]}] (4) to node [right, writes] {$b$} (5);
		\draw[-{Stealth[length=1.4mm, width=1.2mm]}] (5) to node [right, writes] {$a$} (6);
		\draw[-{Stealth[length=1.4mm, width=1.2mm]}] (5) to node [right, writes] {$c$} (7);
        \end{scope}
        
		\end{tikzpicture}
        }
	\end{subfigure}

    \caption{In the top-left corner, the figure shows a $3\times7$ binary matrix $M$ with $n_1 = 1$, $n_2 = 3$, and $n_3 = 2$, and its sequences $\DD$ and $\LL$. 
    Since $\LL$ is not Łukasiewicz, by inverting $M$, we obtain the object~(a) which is not a trie.
    This happens because there are no pending edges to which the pre-order node $5$ can be attached since $\LL[4] = -1$.
    However, the rotated matrix $M^3$ showed in the bottom-left corner produce the valid trie shown in~(b).
    In this case $\LL$ is Łukasiewicz.
    }
    \label{fig1}
\end{figure}

In the following, we prove the main result of this section.

\begin{proof}[Theorem~\ref{theorem number of tries}]
    Let $\sim$ be the relation over $\mathcal{M}$ defined as follows.
    For every $M, \bar M \in \mathcal{M}$ we have that $M \sim \bar M$ holds if and only if, there exists an integer $i$ such that $M^i = \bar M$.
    It is easy to show that $\sim$ is an equivalence relation over $\mathcal{M}$.
    Indeed, $\sim$ satisfies; (i)~reflexivity ($M \sim M$ since $M^0 = M$), (ii)~symmetry (if $M^i = \bar M$, then $\bar M^{(n-i)} = M$), and (iii)~transitivity (if $M^i = \bar M$ and $\bar M^j = \hat M$, then $M^{i+j} = \hat M$).
    By Corollary~\ref{unique matrix in image}, we know that for every $M \in \mathcal{M}$ the equivalence class $[M]_{\sim}$ contains $n$ elements with a unique matrix $\bar{M} \in [M]_{\sim}$ such that $\bar M \in f(\mathcal{U})$.
    Therefore, since $f$ is injective, we deduce that $\lvert \mathcal{U} \rvert = \lvert f(\mathcal{U}) \rvert = \lvert \mathcal{M} \rvert / n$.
    Finally, since it is easy to observe that $\lvert \mathcal{M} \rvert = \prod_{i = 1}^{\sigma}\binom{n}{n_i}$ it follows that $\lvert \mathcal{U} \rvert = \frac{1}{n}\prod_{i = 1}^{\sigma}\binom{n}{n_i}$.
\end{proof}

\section{Entropy of a trie}\label{sec:entropy}

We now use Theorem~\ref{theorem number of tries} from the previous section to introduce the \emph{worst-case entropy} of a trie with given symbol frequencies $n_1,\ldots,n_\sigma$. 

\begin{definition}[Worst-case entropy]\label{def: worst-case}
    The worst-case entropy $\mathcal{H}^{\wc}$ of a trie $\mathcal{T}$ with symbol distribution $n_1, \ldots, n_\sigma$ is $\mathcal{H}^{\wc}(\mathcal{T})  = \log \lvert \mathcal{U} \rvert =  \sum_{i=1}^{\sigma} \log \binom{n}{n_i}  - \log n$
\end{definition}

Clearly the above worst-case entropy is a lower-bound for the minimum number of bits (in the worst-case) required to encode a trie with a given number of symbol occurrences.
Note that $\mathcal{H}^{\wc}(\mathcal{T})$ is always smaller than or equal to the worst-case entropy $\mathcal{C}(n,\sigma) = \log \frac{1}{n} \binom{n\sigma}{n - 1}$~\cite{RRR} of tries having $n$ nodes over an alphabet of size $\sigma$ and it is much smaller when the symbol distribution is skewed.
Next we introduce the notion of \emph{$0$-th order empirical entropy} of a trie $\mathcal{T}$. 

\begin{definition}[$0$-th order empirical entropy]\label{def: 0-entropy}
    The $0$-th order empirical entropy $\mathcal{H}_0$ of a trie $\mathcal{T}$ with symbol distribution $n_1, \ldots, n_\sigma$ is defined as;

    \[
    \mathcal{H}_{0}(\mathcal{T}) = \sum_{i=1}^{\sigma}\frac{n_i}{n}\log\left(\frac{n}{n_i}\right) + \frac{n - n_i}{n}\log\left( \frac{n}{n-n_i} \right)
    \]
\end{definition}

Given bitvector $B_{n,m}$ of size $n$ with $m$ ones, the worst-case entropy is defined as $\mathcal{H}^{\wc}(B_{n,m}) = \log \binom{n}{m}$ while the corresponding $0$-th order empirical entropy is $\mathcal{H}_0(B_{n,m}) = \frac{m}{n}\log(\frac{n}{m}) + \frac{n-m}{n}\log(\frac{n}{n-m})$~\cite{compactDataStructures, elementsOfInformationTheory}.
Furthermore, the following known inequalities relate the above entropies of a binary sequence: it holds that $n\mathcal{H}_0(B_{n,m})-\log(n+1) \leq \mathcal{H}^{\wc}(B_{n,m})\leq n\mathcal{H}_0(B_{n,m})$~\cite[Equation 11.40]{elementsOfInformationTheory}.
The next result is a direct consequence of these inequalities.

\begin{lemma}\label{lemma:emp_wc}
    For every trie $\mathcal{T}$ with $n$ nodes, the following formula holds $n\mathcal{H}_0(\mathcal{T})-\sigma\log(n+1)-\log n \leq \mathcal{H}^{\wc}(\mathcal{T})\leq n\mathcal{H}_0(\mathcal{T})- \log n$
\end{lemma}

The following corollary directly follows from the above lemma.

\begin{corollary}\label{corollary: relation worst-case empirical}
For every trie $\mathcal{T}$ it holds that
$n\mathcal{H}_0(\mathcal{T}) = \mathcal{H}^{\wc}(\mathcal T) + O(\sigma \log n)$
\end{corollary}

Note that the results of Lemma~\ref{lemma:emp_wc} and Corollary~\ref{corollary: relation worst-case empirical} are valid also if we consider the effective alphabet formed by all symbols labeling at least one edge in $\mathcal{T}$, as the other characters do not affect $\mathcal{H}^{\wc}$ or $n\mathcal{H}_0$.
The result of Corollary~\ref{corollary: relation worst-case empirical} is consistent with known results for strings: given a string $S$ of length $n$ over the alphabet $\Sigma$ it is known that $n\mathcal{H}_0(S)=\mathcal{H}^{\wc}(S)+O(\sigma\log n)$ (see~\cite[Section 2.3.2]{compactDataStructures}) where $\mathcal{H}^{\wc}(S)=\log\binom{n}{n_1,\ldots,n_\sigma}$ is the worst-case entropy for the set of strings with character occurrences $n_i$ and ${H}_0(S)$ is the zero-order empirical entropy of $S$ as defined in~\cite{compactDataStructures}.
We now aim to extend the $0$-th order trie entropy to any arbitrary integer $k > 0$.
To this purpose, similarly as in~\cite{xBWTconf, xBWTjournal, rindexTrie}, for every integer $k \geq 0$, we define the \emph{$k$-th order context} of  a node $u$, denoted by $\lambda_k(u)$, as the string formed by the last $k$ symbols of the path from the root to $u$.
Formally, we have $\lambda_0(u) = \epsilon$ and $\lambda_k(u) = \lambda_{k-1}(\pi(u)) \cdot \lambda(u)$.
By definition of $\pi$ (see Section~\ref{sec: Notation}), if a node $u$ has a depth $d$ with $d < k$, then such string is left-padded by the string $\#^{k-d}$.
We now define the integers $n_w$ and $n_{w,c}$, for every $w \in  \Sigma^{k}$ and $c \in \Sigma$ as follows: $n_{w}$ is the number of nodes having $w$ as their $k$-th order context and $n_{w,c}$ is the number of nodes having $w$ as their $k$-th order context and an outgoing edge labeled by $c$.
Formally, $n_w = \lvert \{ v \in V : w = \lambda_k(u) \} \rvert$ and $n_{w,c} = \lvert \{ v \in V : w = \lambda_k(u)\; \text{and}\; c \in out(u)\}\rvert$.
With the next definition we generalize our notion of empirical entropy of a trie to higher orders in a similar way of what has been done in~\cite{xBWTjournal, xBWTconf} for labeled trees.

\begin{definition}[$k$-th order empirical entropy]\label{def: k-entropy}
    For every integer $k \geq 0$, the $k$-th order empirical entropy of a trie $\mathcal{T}$ is defined as follows;

    \[
    \mathcal{H}_{k}(\mathcal{T}) = \sum_{w \in  \Sigma^k}\sum_{c \in \Sigma}\frac{n_{w,c}}{n}\log\left(\frac{n_w}{n_{w,c}}\right) + \frac{n_w - n_{w,c}}{n}\log\left( \frac{n_w}{n_w-n_{w,c}} \right)
    \]
\end{definition}

Note that for $k = 0$ this formula coincides with that of Definition~\ref{def: 0-entropy} since in this case $\Sigma^0= \{\epsilon\}$ and $n_{\epsilon}=n$ and for every character $c_i \in \Sigma$, we have $n_{\epsilon,c_i} = n_i$.
Moreover, analogously to strings, by the log-sum inequality (see \cite[Eq. (2.99)]{elementsOfInformationTheory}), for every trie $\mathcal{T}$ and integer $k\geq 0$ it holds that $H_{k+1}(\mathcal{T}) \leq H_{k}(\mathcal{T})$.

\section{Indexing tries}\label{sec: indexing}

In this section, we prove that it is possible to represent and index a trie in a space upper-bounded by its $k$-th order empirical entropy (see Definition~\ref{def: k-entropy}) plus a $o(n)$ term.
In particular, this representation corresponds to the analogous of the FM-index~\cite{FMindex,FMI-Journal} for tries.
In the following we assume that the alphabet $\Sigma$ is {\em effective}, that is every symbol in $\Sigma$, with the exception of the special character $\#$, labels at least one edge in $\mathcal{T}$ and therefore $\sigma \leq n$.
We denote with $r$ the root of $\mathcal{T}$, which has depth $d=0$ and with $\lambda(r \rightsquigarrow u)$ the string labeling the downward path from $r$ to $u$.
Formally, if $u$ has depth $d$ then $\lambda(r\rightsquigarrow u)=\lambda_d(u)$.
We now extend the total order $\preceq$ over $\Sigma^*$ to the nodes of $\mathcal{T}$ by considering the sequence of nodes $u_1\preceq\ldots\preceq u_n$ sorted according to the co-lexicographic order of the string labeling their incoming path.
Formally, we have that $u_i\preceq u_j$ if and only if $\lambda(r\rightsquigarrow u_i)\preceq \lambda(r\rightsquigarrow u_j)$. Consequently, it is easy to observe that $u_1=r$ holds for every trie $\mathcal{T}$.
We now recall the definition of the $\BWT(\mathcal{T})$ of a trie $\mathcal{T}$~\cite{rindexTrie,xBWTconf,xBWTjournal}.
Informally, $\BWT(\mathcal{T})$ is defined as a sequence of $n$ sets, where the $i$-th set contains the labels outgoing from the $i$-th node according to the co-lexicographic node ordering $\preceq$.

\begin{definition}[XBWT of a trie]\cite[Definition 3.1]{rindexTrie}\label{definition: Burrows-Wheeler}
Given a trie $\mathcal{T}$ with co-lexicographic sorted nodes $u_1,\ldots,u_n$, then $
\BWT(\mathcal{T})=out(u_1),\ldots,out(u_n)$
\end{definition}

This transformation is invertible and can be computed in $O(n)$ time using the algorithm proposed by Ferragina et al.~\cite[Theorem 2]{xBWTjournal}.
In the following we recall the representation of $\BWT(\mathcal{T})=out(u_1)\ldots out(u_n)$ proposed by Belazzougui~\cite[Theorem 2]{Belazzougui2010} which is based on $\sigma$ bitvectors.
This representation was originally introduced to store the trie underlying the Aho-Corasick automaton~\cite{AC75} and later Hon et al.~\cite{Hon2013} observed the connection between this representation and the XBWT.
Given $c\in \Sigma$ we define $B_c[1..n]$ as the bitvector having $B_c[j]=1$ if and only if $c\in out(u_j)$, consequently if $c$ is the $i$-th character of $\Sigma$, then $B_c$ has exactly $n_i$ bits equal to $1$.
We also store the usual array $C[1..\sigma]$ for BWT-based indexes, such that $C[1] = 0$ (corresponding to $\#$) and $C[i] = (\sum_{j = 1}^{i-1} n_j)+1$ for every $i > 1$.
In other words, $C[i]$ stores the number of nodes in $\mathcal{T}$ whose incoming symbol is smaller than $c_i$, where the plus $1$ term is required since the root has no incoming edges.
Consider the alphabet $ \hat \Sigma = \Sigma \setminus\{ \# \}$, given a string $p\in \hat \Sigma^m$, the count {\em subpath query}~\cite{xBWTjournal, rindexTrie} on a labeled tree $\mathcal{T}$ consists in counting the number of nodes reached by pattern $p$, where the path can start from any node of $\mathcal{T}=(V,E)$, i.e., $count(\mathcal{T},p)=\lvert\{u\in V \mid \lambda_m(u)=p \}\rvert$.
In the XBWT the set of nodes reached by $p$ forms an interval in the co-lexicographic sorted sequence $u_1, \ldots, u_n$~\cite{xBWTjournal}, and this interval can be inductively computed by means of {\em forward search} as follows.
Suppose we want to compute the interval $[u_i,u_j]$ for a given pattern $p$ of length $m$, where $p=\alpha c$ with $\alpha \in \hat \Sigma ^{m-1}$ and $c \in \hat \Sigma$.
If $[u_{i'}, u_{j'}]$ is the interval of nodes reached by $\alpha$, then it holds that $i=C[c]+\rank(i'-1,B_c)+1$ and $j=C[c]+\rank(j',B_c)$ where $[u_1,u_n]$ is the interval of the empty pattern $\epsilon$.
We now prove a preliminary result.

\begin{lemma}\label{lemma:trieH0}
Let $\mathcal{T}$ be a trie and $\varepsilon$ be an arbitrary constant with $0 \leq \varepsilon <1$.
If $\sigma \leq n^\varepsilon$, then we can store the $\BWT(\mathcal{T})$ within $n\mathcal{H}_{0}(\mathcal{T})+o(n)$ bits of space and support $count(\mathcal{T}, p)$ queries for a pattern $p \in \hat \Sigma^{m}$ in: 1)~Optimal $O(m)$ time if $\sigma = O(\log^{\varepsilon} n)$ and 2)~near-optimal $O(m\log n)$ time otherwise
\end{lemma}

\begin{proof}
    In order to store and to support rank operations on the $\sigma$ bitvectors $B_c$ corresponding to the $\BWT(\mathcal{T})$, we use the representations of Raman et al.~\cite{RRR} that are summarized in Lemma~\ref{thm:RRR}.
    Using the Solution~\ref{thm:RRR property 1} of Lemma~\ref{thm:RRR}, we store the FID representations of every $B_c$ using $\sum_{i=1}^{\sigma}  \log \binom{n}{n_i} + O(\sigma n\log \log n/\log n)$ bits.
    If $\sigma = O(\log^\varepsilon n)$, then the resulting index takes $\sum_{i=1}^\sigma  \log \binom{n}{n_i} + o(n)$ bits of space. 
    Otherwise as done by Belazzougui~\cite[Theorem 2]{Belazzougui2010}, we use the ID representation~\ref{thm:RRR property 2} of Lemma~\ref{thm:RRR} to represent every $B_c$ within $\sum_{i=1}^\sigma \log \binom{n}{n_i} + o(n) + O(\sigma \log \log n)$ bits of total space which is $\sum_{i=1}^\sigma \log \binom{n}{n_i} + o(n)$ since $\sigma \leq n^\varepsilon$.
    In both cases, in addition to the array $C$, we store an array $A$ of length $\sigma$ such that $A[c]$ stores a pointer to the representation of $B_c$.
    Both arrays take $O(\sigma \log n) = o(n)$ bits of space since $\sigma \leq n^\varepsilon$.
    Note that in the first case we obtain the $O(m)$ time bound as FID representations support rank operations in constant time.
    In the second case, the time complexity follows from the fact that using the ID representation, every forward step takes $O(\log n)$ time due to binary search (see Section~\ref{sec: Notation}).
    Finally, by Definition~\ref{def: worst-case}, both solutions take $\sum_{i=1}^\sigma \log \binom{n}{n_i} + o(n)= \mathcal{H}^{\wc}(\mathcal{T})+o(n)$ bits of space. 
    Therefore, the claimed space bound follows from Lemma~\ref{lemma:emp_wc}.
\end{proof}

The above representation for large alphabets is the one proposed by Belazzougui~\cite[Theorem 2]{Belazzougui2010} to compress the \emph{next} transition function of the Aho-Corasick automaton.
The lemma shows that this representation uses a space bounded by our $0$-th order empirical entropy plus a $o(n)$ term and efficiently supports count queries.
We now analyze the space usage of these representations in terms of $\mathcal{H}^{\wc}(\mathcal{T})$.
Since $\log \binom{m}{n}\geq n$ whenever $n\leq m/2$~\cite[Section 1.1.4]{RRR}, if $n_i \leq n/2$ for every $i \in [\sigma]$, then $\mathcal{H}^{\wc}(\mathcal{T})\geq n-1-\log n$, therefore the $o(n)$ term is also $o(\mathcal{H}^{\wc}(\mathcal{T}))$.
In addition, by Corollary~\ref{corollary: relation worst-case empirical}, we know that $n\mathcal{H}_0(\mathcal{T}) = \mathcal{H}^{\wc}(\mathcal{T}) + O(\sigma \log n) = \mathcal{H}^{\wc}(\mathcal{T}) + o(n)$ assuming $\sigma \leq n^\varepsilon$.
Therefore, we conclude that index of Lemma~\ref{lemma:trieH0} is {\em succinct} since it takes $\mathcal{H}^{\wc}(\mathcal{T}) + o(n) = \mathcal{H}^{\wc}(\mathcal{T}) + o(\mathcal{H}^{\wc}(\mathcal{T}))$ bits of space.

On the other hand, if there exists $n_j > n/2$, since $\sum_{i = 1}^{\sigma} n_i = n-1$ then $j$ has to be unique, i.e., the symbol $c_j$ is the only occurring more than $n/2$ times. We also observe that $\mathcal{H}^{\wc}(\mathcal{T}) = \sum_{i \in [\sigma]}\log \binom{n}{n_i} - \log n = \sum_{\substack{i \in [\sigma] \setminus \{j\}}} \log \binom{n}{n_i} + \log \binom{n}{n - n_j} - \log n$ holds since $\binom{m}{n} = \binom{m}{m-n}$.
As $n-n_j\leq n/2$ we obtain $\mathcal{H}^{\wc}(\mathcal{T})\geq n^*-\log n$ where $n^*=\sum_{\substack{i \in [\sigma] \setminus \{j\}}} n_i + (n-n_j)$.

Now we replace the bitvector $B_{c_j}$ with its complement $\bar{B}_{c_j}$ obtained by replacing in $B_{c_j}$ every $1$ with $0$ and vice versa. Clearly, $\bar{B}_{c_j}$ has $n-n_j\leq n/2$ bits equal to one, and $\rank(i,B_{c_j})=i - \rank(i,\bar{B}_{c_j})$.
Next, as in Lemma~\ref{lemma:trieH0} we represent every bitvector, including $\bar B_{c_j}$, with the ID representation of Lemma~\ref{thm:RRR} within a total space of $\sum_{\substack{i \in [\sigma] \setminus \{j\}}} \log \binom{n}{n_i} + \log \binom{n}{n - n_j}+o(n^*)+O(\sigma\log\log n)$, where $n^*$ is the total number of $1$s in the bitvectors as defined above.
By adding the space for the arrays $C$ and $A$ of Lemma~\ref{lemma:trieH0} we obtain a final space of $\mathcal{H}^{\wc}(\mathcal{T})+o(\mathcal{H}^{\wc}(\mathcal{T})+\log n)+O(\sigma \log n) = \mathcal{H}^{\wc}(\mathcal{T})+o(\mathcal{H}^{\wc}(\mathcal{T}))+O(\sigma \log n)$ bits. Therefore this representation is succinct up to an additive $O(\sigma \log n)$ number of bits.
We finally observe that since $\log \binom{m}{n}\geq \log m$ if $n$ is not $0$ nor $m$, it holds that $\mathcal{H}^{\wc}(\mathcal{T}) =  \Omega( (\sigma-1) \log n)$, and therefore the previous space is also $O(\mathcal{H}^{\wc}(\mathcal{T}))$ assuming $\sigma > 1$.
In this case we can answer the query $count(\mathcal{T}, p)$ in $O(m\log n)$ time. 
Next, we show that the space required for the index of Lemma~\ref{lemma:trieH0} can be further reduced to $n\mathcal{H}_k(\mathcal{T}) + o(n)$ bits when the contexts length $k$ is sufficiently small.
Before doing that, we recall a result from~\cite{FixBlockCompressionBoostJournal} that we will use in the proof of Theorem~\ref{theorem: trieHk}, instantiated for the particular case of binary strings.

\begin{lemma}\cite[Lemma 4]{FixBlockCompressionBoostJournal}\label{lemma: partitionedH0}
    Let $X_1\cdots X_\ell$ be an arbitrary partition of a string $X$ into $\ell$ blocks and let $X_1^b\cdots X^b_m$ be a partition of $X$ into $m$ blocks of size at most $b$, then $\sum_{i=1}^m\lvert X_i^b \rvert \mathcal{H}_0(X_i^b) \leq \sum_{i=1}^\ell\lvert X_i \rvert \mathcal{H}_0(X_i)+(\ell-1)b$
\end{lemma}

Now by using the above Lemma, we show that the fixed block compression~\cite{implicitCompBoost,FixBlockCompressionBoostJournal} used for the higher-order compression of the FM-index can be adapted to tries for the same purpose.

\begin{theorem}\label{theorem: trieHk}
    Let $\mathcal{T}$ be a trie and $\varepsilon$ be an arbitrary constant with $0 \leq \varepsilon <1$.
    If $\sigma \leq n^\varepsilon$, then we can store the $\BWT(\mathcal{T})$ within $n\mathcal{H}_k(\mathcal{T}) + o(n)$ bits of space for every $k = o(\log_\sigma n)$ simultaneously. The index supports $count(\mathcal{T}, p)$ queries for a pattern $p \in \hat\Sigma^{m}$ in optimal $O(m)$ time if $\sigma = O(\log^{\varepsilon} n)$
    and near-optimal $O(m\log n)$ time otherwise.
\end{theorem}

\begin{proof}
    Consider the partition $\mathcal{V}_k$ of $V$ which groups together the nodes of $\mathcal{T}$ reached by the same length-$k$ string/context $w\in \Sigma^k$, that is $\mathcal{V}_k = \{P_w\neq \emptyset \mid w\in \Sigma^k\}$ where $P_w = \{ v\in V \mid \lambda_k(v)=w \}$.
    By Definition~\ref{definition: Burrows-Wheeler}, every part in $\mathcal{V}_k$ corresponds to a range of nodes in the co-lexicographic ordering.
    Given $\mathcal{V}_k=V_1,\ldots V_\ell$, we call $B_{w,c}$ the portion of $B_c$ whose positions correspond to the interval $V_i$ of nodes reached by the context $w \in \Sigma^k$.
    We observe that every $B_{w,c}$ has length $n_{w}$ and contains exactly $n_{w,c}$ bits equal to $1$, where $n_{w}$ and $n_{w,c}$ are those defined in Section~\ref{sec:entropy}.
    As a consequence $\mathcal{H}_k(\mathcal{T})= \frac{1}{n}\sum_{w \in \Sigma^k} n_w\sum_{c \in \Sigma}\frac{n_{w,c}}{n_w}\log (\frac{n_w}{n_{w,c}}) + \frac{n_w - n_{w,c}}{n_w}\log(\frac{n_w}{n_w-n_{w,c}})$ can be rewritten as $\frac{1}{n} \sum_{c\in \Sigma} \sum_{w\in \Sigma^k} \lvert B_{w,c}\rvert \mathcal{H}_0(B_{w,c})$. 
    Therefore, analogously to the case of strings, we can compress every $B_{w,c}$ individually using a $0$-th order binary compressor to achieve exactly $n\mathcal{H}_k(\mathcal{T})$ bits of space.
    Next, we prove that we can approximate this space by considering a fixed-size context-independent partitioning.
    We distinguish two different cases, depending on the alphabet size.
    If $\sigma = O(\log^\varepsilon n)$ we use an idea similar to~\cite{implicitCompBoost} to show that by applying the FID representation of Raman et al.~\cite[Lemma 4.1]{RRR} to every bitvector $B_c$, without any explicit partitioning, we automatically reach the claimed $n\mathcal{H}_k(\mathcal{T})+o(n)$ space bound.
    The FID representation of Raman et al.\ for a bitvector $B$ works as follows.
    Let $n$ be the length of $B$, we divide $B$ into $t=\lceil n/u \rceil$ chunks of fixed size $u = \lfloor \frac{1}{2} \log n\rfloor$. We refer with $B^i$ the $i$-th chunk in $B$ and with $x_i$ the number of $1$s in $B^i$.
    Every chunk is encoded within $\lceil \log \binom{u}{x_i} \rceil \leq u\mathcal{H}_0(B^i) + 1$ bits where the inequality follows from~\cite[Equation 11.40]{elementsOfInformationTheory}.
    Therefore, the space for representing $B$ is at most $\sum_{i=1}^t u\mathcal{H}_0(B^i) + 1$ and the total space to represent the $\sigma$ bitvectors $B_c$ is at most $\sum_{c \in \Sigma} \sum_{i=1}^t u\mathcal{H}_0(B^i_c) + 1  \leq \sum_{c \in \Sigma} \sum_{i=1}^t u\mathcal{H}_0(B^i_c) + O(\sigma n / \log n) = \sum_{c \in \Sigma} \sum_{i=1}^t u\mathcal{H}_0(B^i_c) + o(n)$ since $\sigma = O(\log^\varepsilon n)$.
    Now we upper-bound the left summation in terms of $\mathcal{H}_k(\mathcal{T})$ by applying Lemma~\ref{lemma: partitionedH0} to the above FID partitioning and the optimal one $\mathcal{V}_k$.
    In particular, we obtain that $\sum_{c\in \Sigma} \sum_{i=1}^t \lvert B_{c}^i\rvert \mathcal{H}_0(B^i_c) \leq n\mathcal{H}_k(\mathcal{T}) + \sigma(\ell-1)u$, where $\ell$ is the size of the optimal partitioning $\mathcal{V}_k$. Since the number of contexts in $\Sigma^k$ is at most $\sigma^k$ it is also $\ell \leq\sigma^k$.
    Therefore, we obtain the upper-bound $n\mathcal{H}_k(\mathcal{T})+o(n)$ because the last term $\sigma (\ell-1)u \leq \sigma^{k+1}\lfloor \frac{1}{2} \log n\rfloor$ is $o(n)$ when $k=o(\log_\sigma n)$.
    The FID representations stores in addition other data structures for answering rank/select queries in $O(1)$ time, we note that this machinery requires $O(n \log \log n / \log n)$ bits for every bitvector $B_c$, thus $O(\sigma n \log \log n / \log n)$ of overall additional space, which is $o(n)$ when $\sigma = O(\log^\varepsilon n)$.
    If instead $\sigma = O(\log^\varepsilon n)$ does not hold, as done by Kosolobov and Sivukhin~\cite[Lemma 6]{Kosolobov2019}, we use the explicit fixed block compression of Gog et al.~\cite{FixBlockCompressionBoostJournal} to compress every bitvector $B_c$.
    In particular, we partition every $B_c$ into $t$ blocks $B_c^{i}$ each of fixed size $b = \sigma \log^2 n $ and we apply the ID representation~\ref{thm:RRR property 2} on every $B_c^i$.
    Thus, we obtain a total space in bits which is at most $\sum_{c\in \Sigma}\sum_{i=1}^t \lvert B_{c}^i\rvert \mathcal{H}_0(B^i_c)+ o(x^i_c) + O(\log \log \lvert B_{c}^i\rvert)$, where $x^i_c$ is the number of $1$s in $B_c^i$ and $t=\lceil n/b \rceil$.
    Note that $\sum_{c\in \Sigma}\sum_{i=1}^t o(x^i_c)$ is $o(n)$ since summing over all the $1$s we are counting the number of edges in $\mathcal{T}$.
    Furthermore, the last term $\sum_{c\in \Sigma}\sum_{i=1}^t O(\log \log \lvert B_{c}^i\rvert)$ is at most $O(\sigma \frac{n}{b} \log \log n)=O(\frac{n\log \log n}{\log^2n})=o(n)$.
    So far we proved that the space occupation is at most $\sum_{c\in \Sigma}\sum_{i=1}^t \lvert B_{c}^i\rvert \mathcal{H}_0(B^i_c) + o(n)$.
    Again by Lemma~\ref{lemma: partitionedH0}, we can upper-bound the left summation with $n\mathcal{H}_k(\mathcal{T}) + \sigma(\ell-1)b$ where the last term $\sigma (\ell-1)b \leq \sigma^{k+2}\log^2 n$ is $o(n)$ when $k=o(\log_\sigma n)$.
    For indexing purposes we also store the usual array $R[1..\sigma][1..\lceil n/b\rceil]$ which memorizes the precomputed rank values preceding every block $B_c^{i}$ for every $B_c$.
    Note that using $R$, we can compute $\rank(j, B_c)$ as $R[c][i] + \rank(j - (i-1)b, B_c^{i})$, where $i=\lceil j / b \rceil$ is the index of the interval where position $j$ falls.
    The array $R$ occupies $O(\sigma \frac{n}{b}\log n)=o(n)$ bits of space and we note that within the same space we can store a pointer to every ID representation.
    Note that the partitions are chosen independently of $k$, therefore the theorem holds for every $k=o(\log_\sigma n)$ simultaneously.
\end{proof}

The solution described in the above theorem for large alphabets coincides with the representation proposed by Kosolobov and Sivukhin~\cite[Lemma 6]{Kosolobov2019} to compress the bitvector-based representation of Belazzougui~\cite{Belazzougui2010} using fixed-block compression.
Similarly as before, the above theorem shows that this solution uses a number of bits which is at most our $k$-th order empirical entropy plus a $o(n)$ term.
Now we compare the space bound of Theorem~\ref{theorem: trieHk} with that of the $r$-index for tries proposed by Prezza~\cite{rindexTrie}.
Consider $\BWT(\mathcal{T}) = out(u_1),\ldots,out(u_n)$ (see Definition~\ref{definition: Burrows-Wheeler}).
An integer $i \in [n]$ is a $c$-run break if $c \in out(u_i)$ and either $i = n$ or $c\notin out(u_{i+1})$ holds. 
Then the number $r$ of \emph{BWT runs}, is $r = \sum_{c \in \Sigma} r_c$, where $r_c$ is the total number of $c$-run breaks in $\BWT(\mathcal{T})$~\cite{rindexTrie}.
Prezza showed that his $r$-index supports some operations including count queries for a pattern $p\in \hat\Sigma^m$ in $O(m\log \sigma)$ time using $O(r \log n) + o(n)$ bits of space~\cite[Lemma 4.5]{rindexTrie}.
Moreover, Prezza proved that for every integer $k \geq 0$, the following inequality holds $r \leq \sum_{w \in \Sigma^{k}} \sum_{c \in \Sigma} \log \binom{n_{w}}{n_{w,c}} + \sigma^{k+1}$~\cite[Theorem 3.1]{rindexTrie}, where $n_w$ and $n_{w,c}$ are the integers defined in Section~\ref{sec:entropy}.
By~\cite[Equation 11.40]{elementsOfInformationTheory} this in turn implies that $r \leq n\mathcal{H}_k(\mathcal{T}) +  \sigma^{k+1}$ for every $k$.
In the following proposition, we exhibit a family of trie for which $r = \Theta(n\mathcal{H}_0(\mathcal{T})) = \Theta(n)$ holds.
This proves that the $r$-index of Prezza may use \emph{asymptotically} more space than the index presented in Theorem~\ref{theorem: trieHk} due to the $\log n$ multiplicative factor in the space occupation of the $r$-index.

\begin{restatable}[]{proposition}{propositioni}
There exists an infinite family of tries with $n$ nodes, such that $r=\Theta(n\mathcal{H}_0(\mathcal{T}))$.
\end{restatable}

\begin{proof}
Consider the complete balanced binary tries $\mathcal{T}$ of $n > 1$ nodes over the binary alphabet $\Sigma = \{a, b\}$.
We observe that if $h$ is the height of $\mathcal{T}$, then the set of strings spelled in $\mathcal{T}$, considering also those reaching internal nodes, is $\bigcup_{i=0}^h \Sigma^i$ and each of these strings reaches a single node.
In this case, we have $n_1 = n_2 = (n - 1)/2$ since half edges are labeled by an $a$ and the remaining by $b$.
Consequently, by Definition~\ref{def: 0-entropy}, we have that $n\mathcal{H}_0(\mathcal{T}) = (n-1)\log(\frac{2n}{n-1}) + (n+1)\log(\frac{2n}{n+1}) \leq 2n$ and more precisely $n\mathcal{H}_0(\mathcal{T}) \approx 2n$.
Therefore, to prove the proposition, we show that $r = \Theta(n)$.
Consider the nodes $u_1,\ldots, u_n$ in co-lexicographic order, we have that $out(u_i)=\{a,b\}$ if $u_i$ is an internal node, and obviously $out(u_i)=\emptyset$ if $u_i$ is a leaf.
Since $n > 1$, node $u_1$ is the root and internal, while $u_n$ is necessarily the leaf reached by the string $b^{h}$.
Consequently, we have that $r$ is exactly twice the numbers of leaf runs in $\BWT(\mathcal{T})$.
We now show that every maximal run of leaves contains exactly two nodes.
Consider the leaf $u_i$ reached by the string $a\alpha$ with $\alpha\in\Sigma^{h-1}$, clearly the next node $u_{i+1}$ is the leaf reached by the string $b\alpha$.
We can also observe that the node $u_{i-1}$ is necessarily the one reached by $\alpha$, which is internal since $\alpha \in \Sigma^{h-1}$. 
The string reaching node $u_{i+2}$ is $b\alpha[j+1..h-1]$ where $j$ is the position of the leftmost $a$ in $\alpha$.
Note that if $j$ does not exist then $\alpha=b^{h-1}$ and therefore $u_{i+1}$ is $u_n$, otherwise $b\alpha[j+1..h-1]$ has length at most $h-1$ and therefore $u_{i+2}$ is an internal node.
Since the number of leaves is $(n + 1)/2$, the number of leaf runs is $(n + 1)/4$ and therefore $r=\Theta(n)$.
\end{proof}

Moreover, both indices can retrieve the co-lexicographic index of the child node reached by a label $c \in \Sigma$ by executing a single forward step.
This can be used to perform prefix searches.
The $r$-index also supports the navigation query of retrieving the $i$-th child of any $j$-th co-lexicographic node in $O(\log^2 \sigma)$ time, while our representation can trivially support this query in $O(\sigma)$ time by finding the $i$-th smallest $c$ such that $B_c[j]=1$ and then executing a forward step.
We leave as an open problem the possibility of improving the time complexity for this operation, which is important to visit the trie in the compressed representation.


\end{document}